\documentclass[conference]{IEEEtran}
\usepackage{amsthm}
\usepackage{braket}
\usepackage{braket}
\usepackage{balance}
\usepackage{mathtools}
\usepackage{cite}
\usepackage{amsmath,amssymb,amsfonts}
\usepackage{algorithm2e}
\usepackage{graphicx}
\usepackage{textcomp}
\newtheorem{lemma}{\textbf{Lemma}}

\usepackage{xcolor}
\usepackage{ulem}
\usepackage{cite}
\usepackage{amsmath,amssymb,amsfonts}
\usepackage{graphicx}
\usepackage{textcomp}
\makeatletter
\renewcommand*\env@matrix[1][*\c@MaxMatrixCols c]{%
  \hskip -\arraycolsep
  \let\@ifnextchar\new@ifnextchar
  \array{#1}}
\makeatother

\usepackage{braket}
\usepackage{xcolor}
\usepackage{ulem}
\usepackage{algorithm2e}
\usepackage{braket}
\usepackage{balance}
\usepackage{mathtools}

\newcommand{\ketbra}[2]{\vert #1 \rangle  \langle #2 \vert}

\newtheorem{theorem}{\textbf{Theorem}}
\newtheorem{proposition}{\textbf{Proposition}}
\newtheorem{corollary}{\textbf{Corollary}}
\def\BibTeX{{\rm B\kern-.05em{\sc i\kern-.025em b}\kern-.08em
    T\kern-.1667em\lower.7ex\hbox{E}\kern-.125emX}}
\usepackage[T1]{fontenc}
\usepackage{hyperref}

\def\BibTeX{{\rm B\kern-.05em{\sc i\kern-.025em b}\kern-.08em
    T\kern-.1667em\lower.7ex\hbox{E}\kern-.125emX}}
\begin{document}

\title{The Quantum Internet: an Efficient Stabilizer states Distribution Scheme
}

\author{
    \IEEEauthorblockN{
        Seid Koudia\IEEEauthorrefmark{3}
        \vspace{6pt}
    }

    \IEEEauthorblockA{
        \IEEEauthorrefmark{3}\textit{Department of Physics ``Ettore Pancini''}\\
        \textit{University of Naples Federico II}\\
        Naples, 80125 Italy\\
        {seid.koudia@unina.it}
        \vspace{6pt}
    }

}

\maketitle

\begin{abstract}
Quantum networks constitute a major part of quantum technologies. They will boost distributed quantum computing drastically by providing a scalable modular architecture of quantum chips, or by establishing an infrastructure for measurement based quantum computing. Moreover, they will provide the backbone of the future quantum internet, allowing for high margins of security. Interestingly, the advantages that the quantum networks would provide for communications, rely on entanglement distribution, which suffers from high latency in protocols based on Bell pair distribution and bipartite entanglement swapping. Moreover, the designed algorithms for multipartite entanglement routing suffer from intractability issues making them unsolvable exactly in polynomial time.  In this paper, we  investigate a new approach for graph states distribution in quantum networks relying inherently on local quantum coding --LQC-- isometries and on multipartite states transfer. Additionally, single-shot bounds for stabilizer states distribution are provided. Analogously to network coding, these bounds are shown to be achievable if appropriate isometries/stabilizer codes in relay nodes are chosen, which induces a lower latency entanglement distribution. As a matter of fact, the advantages of the protocol for different figures of merit of the network are provided. 
\end{abstract}


\section{Introduction}
Quantum networks \cite{rodney-book,bassoli2021quantum} lie at the heart of many future quantum technologies, and are the basis of the quantum internet \cite{Pirandola2016PhysicsUT,cacciapuoti2022quantum,gyongyosi2022advances}. They rely on genuinely quantum effects, allowing them to overcome classical approaches, or allowing them to achieve unprecedented tasks \cite{seid,koudia2021causal,miguel2021optimized,miguel2021genuine}.  Indeed, quantum networks promise a better security of information exchange among parties through different quantum key distribution protocols, as well as other tasks suck as quantum secret sharing, quantum leader election,  distributed quantum computing and quantum sensing \cite{Wehner2018QuantumIA}. 

Quantum multipartite entanglement is the backbone of most of the applications.  In particular, for many party-protocols, multipartite states are targeted, in order to fulfill the promises of the functionalities of the underlying quantum network \cite{hahn2019quantum}. One of the most important classes that allow for the achievement of the advantages in many applications, as well as providing a fluid connectivity in the network are \textit{graph states}.  These states are famously known in the design of quantum error correcting codes under the name of \textit{stabilizer states} \cite{hein2006entanglement}. This is an important class of entangled states, not only for its use to connect different quantum nodes for the establishment of a modular architecture for distributed quantum computing, but it also allows for new types of secure quantum computing either in the cloud, i.e., as blind quantum computing, or in a distributed way, i.e, as in measurement based quantum computing.

Nevertheless, the main obstacle that faces quantum networks --in the first place-- is the distribution of these resourceful multipartite entangled states \cite{hahn2022limitations, miguel2021optimized}. In the one hand, Different protocols in the literature have been studying the problem of entanglement distribution in quantum networks from different perspectives. Yet, many of the known protocols rely on the intermediate distribution of instances of Bell pair states, which, after different local operations on repeater nodes, establishes links between the clients and a central node  which teleports the target graph state \cite{omar,avis2022analysis}. A typical repeater node in today’s quantum networks is capable of a single effective functionality: forwarding -swapping- bipartite entanglement. But there is no intrinsic reason why we must assume this is the only function ever permitted to nodes and, in application-level overlay networks and multihop  networks, for example, allowing nodes to have a wider variety of functions makes sense. Clearly, this nature of quantum repeaters poses a scalability issue of the proposed protocols.

Moreover, the designed entanglement routing protocols, so far, deal mainly with point-to-point entanglement establishment \cite{van2013path,pant2019routing,gyongyosi2020dynamics,Pirandola2019EndtoendCO}. Indeed, future concrete application of quantum networks, should deal with distributing multipartite entangled states, rather than solely the end-to-end Bell pair distribution.  Many efforts have been put forward in this direction, although they suffer from major problems. In contrast to the end-to-end entanglement routing which might be solved with algorithms that are a variety of Dijkstra's algorithm -which is able to find an optimal route in polynomial time in the number of nodes in the network-, multipartite entanglement distribution protocols suffer from the fact that they rely on NP-hard problems such as the minimum spanning tree or the Steiner tree. As a matter of fact, these can be approximated by algorithms that run in a polynomial time but they are never exact.

In the other hand, harnessing the full capacities of communication networks should be taken into account when studying entanglement distribution. Namely, constrained quantum networks should be considered for entanglement distribution, when the use of different point-to-point channels building the network is limited, or when the time slots allowed for the entanglement distribution required by the clients are fixed. These assumptions are more practical for many reasons. Mainly, this is because the quantum hardware, such as quantum memories, are expensive to be taken unlimited \cite{zhao2009millisecond}. Additionally, it is more practical to consider relatively lower time-slots of entanglement distribution to increase the communication data rate and decrease the latency, or to make different protocols more robust to adversarial attacks on security schemes relying on multipartite entanglement.


\subsection{Previous Work}
Quantum entanglement routing has been the focus of many works in the past few years. A design of an optimal routing scheme has been presented in \cite{marcello}, with the end-to-end entanglement rate was set as the optimality figure of merit. In \cite{Moe}, the authors have presented a remote entanglement distribution scheme for linear repeater chains, where an optimal route maximizing the end-to-end entanglement is determined alongside with the optimal sequence of the entanglement swaps.  Distributed entanglement routing algorithms, with latency taken into account has been the purpose of \cite{dahlberg}. The problem of optimizing end-to-end entanglement in a many source-destination scenario has been studied in details as a multicommodity flow problem in \cite{wehner-multico}. In \cite{Gyongyosi_2019}, a decentralized adaptive routing scheme has been developed, in which the imperfection of quantum memories is taken into account.  A multipath routing approach for multiple end-to-end entanglement has been thoroughly investigated in \cite{Pirandola_multi}.

In \cite{Rodeny-encoding}, the authors handled the case where quantum repeaters are allowed to perform quantum encoding, where it has been shown that the later increases drastically the end-to-end entanglement rates with respect to usual protocols.  This has been recently experimentally investigated in NV centers in \cite{jin-mohsen}. The effect of this type of intermediate encoding on end-to-end key rate generation in QKD has been studied in \cite{ranu2022qkd}. 

For Multipartite entanglement generation and distribution in quantum networks relying on a central node connected by EPR pairs to the remote clients has been the focus of \cite{wehner-central}. In the same spririt, multipartite entanglement distribution relying on finding an optimal spanning tree of the remote client nodes has been investigated in \cite{omar}. Graph states distribution relying on finding  the Steiner tree spanning the remote nodes, and the use of local graphical operations on remote nodes has been elaborated in \cite{markham}. Similar Graph states distribution protocols with lower EPR resource consumption have been given in \cite{towsley2,Towsley1}. Protocols relying on an inherentely multipartite approach for multipartite entanglement distribution have been designed in \cite{dur-multi}. More related to our paper, with a different approach, is \cite{dagmar} where repeater nodes in the quantum network are allowed to perform quantum coding.

\subsection{Outline and Contribution}
In the present paper, we consider quantum intermediate relay nodes where the contents of outgoing states are  causal functions of the contents of received states, referring to this mapping as Local Quantum Coding. First,  we show that Local Quantum Coding --LQC-- in intermediate nodes, can be used to distribute the particular class of graph states in a distributed way among remote client nodes. This feature will be proved to be important for the following reasons:
\begin{itemize}
    \item It is useful to beat the scalability issues of the previous protocols.
    \item Providing a distributed architecture for multipartite entanglement distirbution
\end{itemize}
Second, we show that the scheme of using LQC, is nothing but a particular network coding strategy \cite{yeung,muriel,bassoli2013network}. This allows to achieve the single-shot capacity of the appropriate quantum network topologies, which normal entanglement routing fails to fulfill.  To this aim:
\begin{itemize}
    \item  Different bounds for stabilizer states distribution capabilities of the underlying network would be obtained in terms of min-cuts capacities corresponding to different bipartitions of the remote client nodes requiring the graph state.
\end{itemize}
By considering the quantum network as a tensor network with a particular instance of tensors being specific coding isometries, the advantages of the advantages of our scheme are given in terms of the following figure of merits and applications:
\begin{itemize}
    \item We show that by allowing the intermediate nodes in the quantum network to use LQC, an exponential latency reduction in the entanglement distribution time is achieved.
    \item We show that the use of LQC reduces drastically the memory qubits overhead of the entanglement distribution cost.
    \item We discuss how by allowing intermediate use of quantum error correcting codes, distributed storage can be achieved in quantum networks.
\end{itemize}

The paper is structured as follows. In Section.~\ref{sec2} some preliminaries are be given. In particular, graph states and their entanglement structure will be highlighted. Additionally, the concept of tensor networks along with the notion of min-cut will be presented. In Section.~\ref{sec3}, our model of the distribution quantum network will be elaborated. An equivalence between our model and a network of stabilizer states contracted by Bell pairs is given in Section.~\ref{sec:04}. In Section.~\ref{sec:05} our results comprising the different bounds on the consumed resources to single-shot distribute a target graph state will be established. Section.~\ref{sec:06} illustrates the obtained results and bounds for some specific network topologies. In particular, the protocol will be benchmarked to Bell pair based protocols with the time complexity of the graph states distribution and the number of memory qubits taken as figure of merits. We also study the application of choosing the LQC schemes to be valid stabilizer quantum error correcting schemes in order to achieve distributed quantum storage within quantum networks. We finish the paper with conclusions and future work in Section.~\ref{sec:07}.

\section{Preliminaries} \label{sec2}
In this section, we provide the necessary preliminaries and tools allowing for the understanding of the derived results. 
\subsection{Graph states}
Graph states admit a simple description in terms of mathematical graphs \cite{hein2006entanglement}.  
A graph  is a pair $(G,V)$ of a finite set $V=\{1.\dots,N\}$ and a set $E\subset [V]^2$. The elements of $V$ are called the vertices of the graph and the elements of $E$ are called its edges. Additionaly, we denote by $N_a$ the set of neighboring vertices of the vertex $a$.

The graph state $\ket{G}$ that corresponds to the graph $G$ is the pure state given as 
\begin{equation}
    \ket{G}=\prod_{\{a,b\}\in E}U_{ab} \ket{+}^V
\end{equation}
where 
\begin{equation}
    \ket{+}^V=\otimes_{a\in V} \frac{1}{\sqrt{2}}(\ket{0}+\ket{1})^a 
\end{equation}
Similarly, $U_{ab}$ is the controlled $Z$ gate between the vertices $a$ and $b$.

It is known that any graph state is equivalent to a stabilizer state under a certain class of local operations, the so-called local Clifford operations \cite{hein2006entanglement}, which map the set of stabilizer states onto itself. Accordingly, graph states can be defined uniquely as the common eigenvector with eigenvalues equal to one  in $\mathbb{C}^V$ to the set of independent commuting operators 
\begin{equation}
    K_a=\sigma_x^a\sigma_z^{N_a}
\end{equation}
for all $a\in V$. 
\subsection{Entanglement structure of graph states}
Generally, the many interesting entanglement structure properties of any multipartite state are captured by the reduced state after tracing out parts of the system. Specifically, for a given pure multipartite state $\ket{\Psi}$. For a specific bipartition of the state, namely $\{A,B\}$, the reduced state on $A$
\begin{equation}
    \rho^A_{\Psi}=\mathrm{Tr}_B(\rho)
\end{equation} 
where $\rho=\ket{\Psi}\bra{\Psi}_{AB}$. The entanglement rank of the state $\ket{\Psi}$ in the bipartition $\{A,B\}$, quantifies the amount of Bell pairs needed to create such bipartition, and is given as \cite{hein2006entanglement}:
\begin{equation}
r_A(\ket{\Psi})=\log_2[rank(\rho^A_{\Psi})] \label{Eq:08}
\end{equation}

\subsection{The min-cut}
For a network $(G,V,S,T)$  a cut $(\tilde{S},\tilde{T})$ is a partition $V=\tilde{S}\cup \tilde{T}$ such that $\tilde{S}\subset S $ and $\tilde{T} \subset T$. The min-cut $\mathrm{MC}$ associated with the network is defined to be the minimum over all possible cuts of the product \cite{cui2016quantum,yeung}
\begin{equation}
    \prod_{(u,v)\in E, u\in \tilde{S}, v \in \tilde{T}} d_{(u,v)}
\end{equation}
Namely,
\begin{equation}
    \text{MC}=\min_{(\tilde{S},\tilde{T})}\prod_{(u,v)\in E, u\in \tilde{S}, v \in \tilde{T}} d_{(u,v)}
\end{equation}
It is much more convenient in the remaining of the paper, to work with the min-cut with respect to $\log_2$ of the dimensions, hence the min-cut would be given by 
\begin{equation}
    \text{MC}=\min_{(\tilde{S},\tilde{T})}\sum_{(u,v)\in E, u\in \tilde{S}, v \in \tilde{T}} \log_2(d_{(u,v)})
\end{equation}
Working with this definition of the min-cut would allow us to work later with the entanglement ranks instead of the dimensions of the states distributed. 

\section{The Model} \label{sec3}
\begin{figure}[t]
    \begin{center}
    \label{Fig:1}
        \includegraphics[width=0.9\columnwidth]{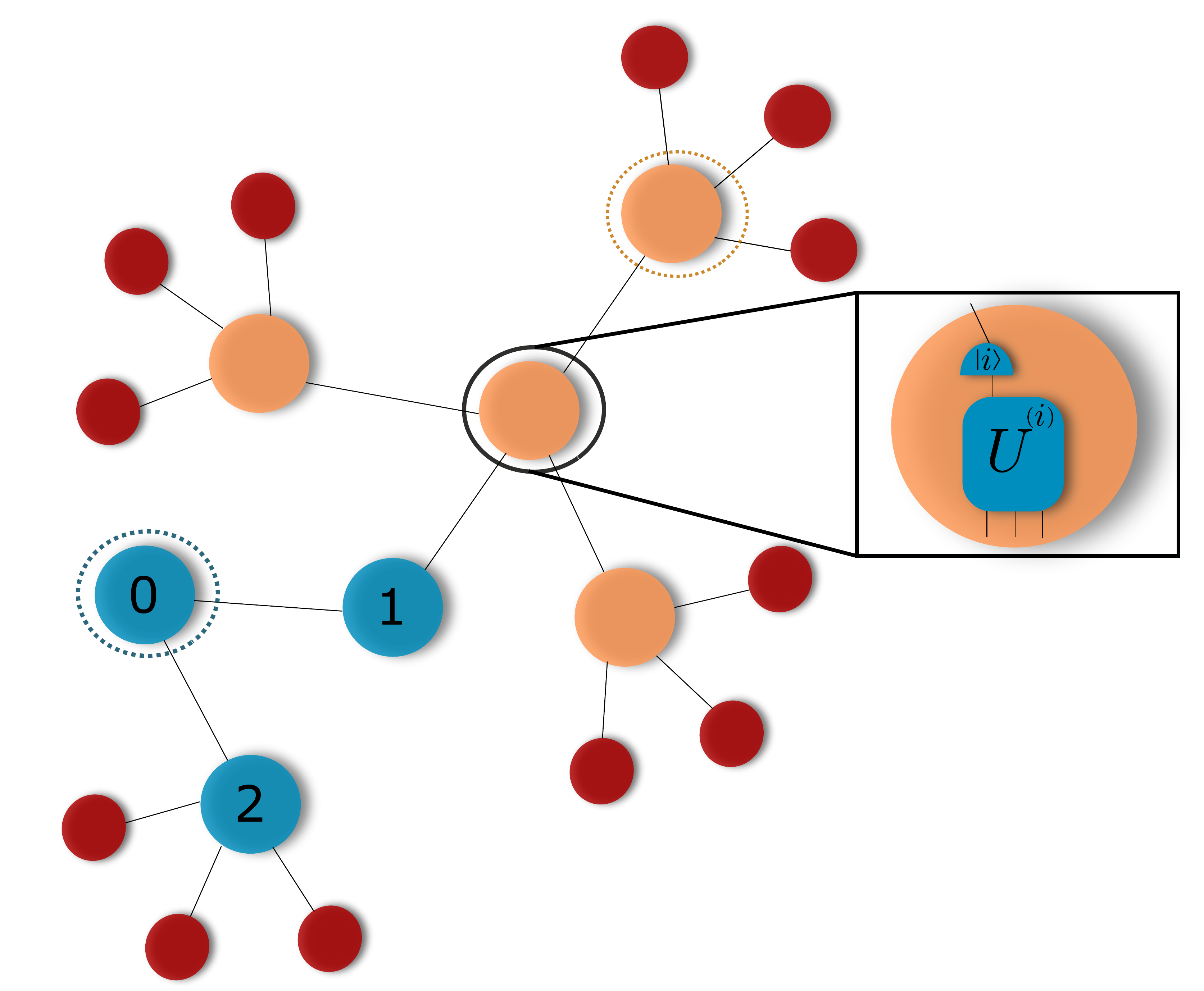}
        \end{center}       
    \caption{Scheme for distributing multipartite entanglement states in a quantum network. The network is composed by twelve clients. each served by a \textit{relay node}. To the right is a magnification of a relay node, which exploits an LQC $S=\sum_{i=1}^{d} U^{(i)}\otimes \ketbra{i}{i}$ to distribute multi-partite entanglement to remote nodes.}
\end{figure}

In our model of quantum networks we are embedding all relay nodes with the possibility of establishing an LQC strategy to generate locally graph states, instead of just doing swapping operations and joint bipartite measurements. Formally, each node $a\in V$ is able to generate a graph state $\ket{G_a}$:
\begin{equation}
\label{Eq:10}
     \ket{G_a}=\sum_{i=1}^{d} U^{(i)}_{a}\ket{0}^{n_a} 
 \end{equation}
following a coherent control strategy according to the isometry:
\begin{equation}
\label{Eq:14}
    S_{a}=\sum_{i=1}^{d} U^{(i)}_{a}\otimes \ketbra{i}{i}_c
\end{equation}
where $d$ is the dimension of the control degree of freedom. Indeed, after measurement of the control degree of freedom in the Fourier basis, the state in Eq.~\ref{Eq:10}, as well as local unitary equivalent to it, are post-selected.

Each party of the obtained graph state $\ket{G_a}$, locally generated at node $a$,  would be distributed to the neighboring relay nodes $\{N_a\}$. Each part serves as a control degree of freedom, in its turn, of an appropriate LQC process $\sum_i  U^i_{b} \otimes \ketbra{i}{i}$, in the corresponding neighboring node $b$. This is illustrated in Fig.~1.  A more specific scheme relying on indefinite causal ordering to generate entanglement has been established and analyzed in details  in \cite{koudia2021causal}. 

In the next section, we show that by considering a network whose nodes $\{a\}_{a\in V}$ are endowed with coding isometries $\{S_a\}_{a\in V}$, indeed, a graph state can be distributed appropriately among the sink nodes, by studying the equivalence to a dual model of the network.

\section{Distributed Graph States Generation}\label{sec:04}

In this section, we show formally, how the endowment of quantum networks with chosen LQC in its relay nodes, allows for the distributed generation of graph states.

\begin{lemma}
\label{lemma1}
Let a quantum network whose relay nodes are endowed with  controlled-isometries. The isometries $\{S_x\}_{x\in v}$ are assumed to be able to generate locally graph states in an appropriate way, upon a control degree of freedom received from the immediate previous nodes $x\in N_y$ according to 
\begin{equation}
\label{Eq:14}
    S_{y}=\sum_{i=1}^{2^r} U^{(i)}_{y}\otimes \ketbra{i}{i}_x
\end{equation}
with $r$ being the number of --incoming-- control qubits. 
This model of the quantum network is equivalent to a network whose vertices $\{x\}$ are endowed with a graph state $\{\ket{G_x}\}$. Each state of these is contracted with the neighbooring state according to the edges $\{e\}$ which represent EPR pairs, as 
\begin{equation}
  \bra{e_{xy}}\ket{G_x}\otimes\ket{\Bar{G}_y}
\end{equation}
where $\ket{\Bar{G}_y}$ is the augmented graph state given by
\begin{equation}
    \ket{\Bar{G}_y}=\ket{0}^{\Bar{a}_y}\ket{0}^{a_y}\ket{G_y/a_y}+\ket{1}^{\Bar{a}_y}\ket{1}^{a_y}\sigma_Z^{N_{a_y}}\ket{G_y/a_y}
\end{equation}
for all $a_y\in V_y$. 
\begin{proof}

In what follows we give a sketch of the proof. 
In general any two graph states $\ket{G_x}$ and $\ket{G_y}$ have the  following schmidt decomposition on single qubit bipartitions \cite{hein2006entanglement}:
\begin{align}
    \ket{G_x}&=\ket{0}^{a_x}\ket{G_x/a_x}+\ket{1}^{a_x}\sigma_Z^{N_{a_x}}\ket{G_x/a_x}\nonumber\\
    \ket{G_y}&=\ket{0}^{a_y}\ket{G_y/a_y}+\ket{1}^{a_y}\sigma_Z^{N_{a_y}}\ket{G_y/a_y}
\end{align}
with $\braket{G_i/a_i|\sigma_Z^{N_{a_i}}|G_i/a_i}=0$ if the underlying graph $G_i$ is connected. 
Similarly, an augmentation graph state $\ket{\bar{G}_y}$ of the state $\ket{G_y}$ on node $a_y$ has the the decomposition 
\begin{equation}
    \ket{\bar{G}_y}=\ket{0}^{\Bar{a}_y}\ket{0}^{a_y}\ket{G_y/a_y}+\ket{1}^{\Bar{a}_y}\ket{1}^{a_y}\sigma_Z^{N_{a_y}}\ket{G_y/a_y}
\end{equation}
On the one hand, by contracting the two states with a Bell pair, the resulting state would be given by 
\begin{equation}
    \ket{G_x/a_x}\ket{0}^{a_y}\ket{G_y/a_y}+\sigma_Z^{N_{a_x}}\ket{G_x/a_x}\ket{1}^{a_y}\sigma_Z^{N_{a_y}}\ket{G_y/a_y}
\end{equation}

On the other hand, we let one qubit from $\ket{G_x}$ to be sent to node $y$ to control a LQC strategy $S_y$ given as 
\begin{equation}
    S_y=U_y^0\otimes \ketbra{0}{0}^x+U_y^1\otimes \ketbra{1}{1}^x\label{eq:18}
\end{equation}
where
\begin{align}
    U_y^0\ket{0}^{\otimes n_y}&=\ket{0}^{a_y}\ket{G_y/a_y}\nonumber\\
    U_y^1\ket{0}^{\otimes n_y}&=\ket{1}^{a_y}\sigma_Z^{N_{a_y}}\ket{G_y/a_y}
\end{align}
generating the graph state $\ket{G_y}$ after measuring the distributed qubit from $x$ in the coherent basis. Indeed, the overall state obtained before measuring the control is given by 
\begin{equation}
    \ket{G_x/a_x}\ket{0}^{a_y}\ket{G_y/a_y}+\sigma_Z^{N_{a_x}}\ket{G_x/a_x}\ket{1}^{a_y}\sigma_Z^{N_{a_y}}\ket{G_y/a_y}
\end{equation}
The procedure extends easily to $r$ incoming qubits, in the same spirit,  leading to the equivalence of the two approaches. 
\end{proof}
\end{lemma}
To get a feel of this lemma lets suppose that the relay nodes $x, y$ choose appropriate LQC strategies enabling the generation of $GHZ$ states $\ket{GHZ}_x$ and $\ket{GHZ}_y$ respectively. By sending a single qubit  from $\ket{GHZ}_x$  in order to control the LQC in $y$, the overall state is given up to a local unitary by 
\begin{equation}
    \frac{1}{\sqrt{2}}(\ket{00}_x\ket{000}_y+\ket{11}_x\ket{111}_y)
\end{equation}
it is easy to notice that this state is equivalent to the state $\bra{EPR}_{xy}\ket{GHZ}_x\ket{GHZ}_y$. 
\begin{theorem}
\label{Theorem:04}
The state distributed between the clients in the quantum network where the relay nodes use appropriate LQC strategies of single qubit unitaries is given by
\begin{equation}
\ket{\Psi}=\otimes_x^{|V|}\bra{G_x}\otimes_e^{|E|}\ket{e} \label{Eq:theorem 4} 
\end{equation}
This overall state  is a graph state.
\begin{proof}
The state in \ref{Eq:theorem 4} results directly from a direct application of Lemma.~\ref{lemma1} to many contractions between the relay nodes leading to the overall state $\ket{\Psi}$ between the clients. Therefore, we only give the proof for the statement about the fact that the over state is a graph state.  

First we should note that any graph state is Local Clifford equivalent to a stabilizer state. Hence, in what comes we will limit ourselves to stabilizer states.  
We should point out that any stabilizer state is characterized by the set of generators of its stabilizer group. These are elements of the local Pauli group that mutually commute, forming an abelian group. Let $\ket{\Phi}_{AB}$ be a stabilizer  state whose stabilizer group is $S$. Similarity, Let $\ket{\chi}_A$ be a stabilizer state with corresponding stabilizer group $T$. The contraction of the two states $\ket{\Psi}_B=\braket{\chi|\Phi}_B$ can be written in density operator notation as
\begin{align}
    \ketbra{\Psi}{\Psi}_B&=\bra{\chi}_A\ketbra{\Phi}{\Phi}_{AB}\ket{\chi}_A 
\end{align}
By using properties of the trace operator and of pure density matrices, we can write the density operator  $\ketbra{\Psi}{\Psi}_B$ in a trace form as:
\begin{equation}
   \ketbra{\Psi}{\Psi}_B=\mathrm{Tr}(\ketbra{\chi}{\chi}_A\ketbra{\Phi}{\Phi}_{AB}) 
\end{equation}
To simlify the trace formula, we can use the fact that the density matrices of the stabilizer states $\ket{\chi}_A$ and $\ket{\Phi}_{AB}$ might be written respectively as: 
\begin{align}
    \ketbra{\chi}{\chi}&=\frac{1}{|S|}\sum_{s_{AB}\in S} {s_{AB}}\nonumber\\
    \ketbra{\Phi}{\Phi}&=\frac{1}{|T|}\sum_{t_{A}\in T} {t_{A}}
\end{align}
Therefore, we have 
\begin{equation}
    \ketbra{\Psi}{\Psi}_B=\sum_{s_{AB}}\sum_{t_{A}} \frac{1}{|S||T|}\mathrm{Tr}(s_{AB}t_A)
\end{equation}
The trace in the last equality is non-vanishing if and only if 
\begin{equation}
    s_{AB}=t_A\otimes u_B
\end{equation}
This is due to the fact that stabilizer group elements are a product of traceless Pauli matrices, which make them traceless in their turn.  
This means that the sum is carried only over the elements $s_{AB}$ of the form $t_A\otimes u_B$. It is worthnoting that both $t_A$ and  $u_B$ are both Local Pauli operators with support on $A$ and $B$ respectively.  By this, we conclude that:  
\begin{equation}
\label{reducedgraph}
   \ketbra{\Psi}{\Psi}_B = \frac{|K|\cdot2^{|A|} }{|S||T|}\sum_{u_{B}}u_{B}
\end{equation}
The factor $2^{|A|}$ is due to the fact that $\mathrm{Tr}(\mathbb{I}^A)=2^{|A|}$, where $\mathbb{I}^A$ is the identity on the subspace $A$. 
Furthermore, because all the state $\ket{G_x}$ or $\ket{e}$ are graph states, they should have full rank stabilizer groups, therefore the full cardinally of the stabilizers $S$ and $T$ should be respectively given by:
\begin{align}
    |S|&=2^{|A|+|B|}\nonumber\\
    |T|&=2^{|A|}
\end{align}
By plugging this into Eq.~\ref{reducedgraph} we obtain:
\begin{align}
    \ketbra{\Psi}{\Psi}_B &= \frac{|K|\cdot2^{|A|} }{2^{2|A|+|B|}}\sum_{u_{B}}u_{B}\nonumber\\
    &=\frac{|K| }{2^{|A|+|B|}}\sum_{u_{B}}u_{B}
\end{align}
Now we can notice that in order for $\ketbra{\Psi}{\Psi}_B$ to be a graph state, the elements $\{u_B\}$ should form a group and satisfy 
\begin{align}
    &|\{u_B\}|=2^{B}\nonumber\\
    &[u_b,u'_b]=0 \quad \forall u_B,u'_B \in \{u_B\}
\end{align}
The first condition is to have a full rank stabilizer group, and the second guarantees that the group is indeed abelian.  
The first condition requires that $|K|$ should satisfy:
\begin{equation}
    |K|=2^{|A|}
\end{equation}

In order to complete the proof, we should whow that $\{u_B\}$ form an abelian group of local Pauli operators, making the state $\ket{\Psi}_B$ a stabilizer state. \\
As a matter of fact, we note that since $\{t_A\otimes u_B\}$ are elements of an abelian group, they should pairwise commute.  In order to prove that the operators $\{t_A\otimes u_B\}$ form a group we note that for any two elements $t^1_A\otimes u^1_B$ and $t^2_A\otimes u^2_B$, their product is of the form   $t^1_At^2_A\otimes u^1_Bu^2_B$, hence of the form $t_A\otimes u_B$. This proves the closure of the set. It is easy to note that any operator in this set is its own inverse, hence the set of operators $\{t_A\otimes u_B\}$ is an abelian group.
In turn, since $\{t_A\}$ form an abelian group, similarly, $\{t_A\otimes u_B\}$ form abelian subgroup of the stabilizer $S$,  $\{u_B\}$ should form an abelian group as well. Consequently, the abelian group $\{u_B\}$ should  describe a stabilizer state. By considering $\ket{\Phi}_{AB}$ in \ref{Eq:theorem 4} to be the stabilizer state given by $\otimes_e^{|E|}\ket{e}$ and $\ket{\chi}_A$ to be the stabilizer state given by $\otimes_x^{|V|}\ket{x}$, and by applying the previous discussion, we conclude that the state $\ket{\Psi}$ shared between the clients is indeed a graph state. 
\end{proof}

\end{theorem}
It is worth-mentioning that the proof of \textbf{Theorem.~\ref{Theorem:04}}, is device-independent, as it is solely relying on the group theoretical structure of stabilizer states. Thus, the proof holds independently of the technological platform implementing the LQC generating entanglement.

\section{Single-Shot network Capacities}\label{sec:05}
In this section, we establish entanglement distribution bounds in quantum networks using network coding methods. we quantify the different resources needed for the Establishment of target graph states with known entanglement ranks.

\subsection{Entanglement Distribution Bounds}
\begin{theorem}
\label{Theorem:02}
A graph state with a list of bipartite entanglement ranks $\{r_{AB}\}$ can be distributed in a quantum network by a single use of the latter if and only if:
\begin{equation}
    MC_{\{A,B\}}\geq r_{AB}=\text{rank}_{\mathbb{F}_2}(\Gamma_{AB})=\log_2[\text{rank}(\rho^G_A)] \label{Eq:15}
\end{equation}
for every bipartition $\{A,B\}$ of the target graph state. 
\begin{proof}
We are gonna prove this theorem by contradiction. 
A bipartition of the state of the client nodes into $\{A,B\}$ would induce a bipartition of the network into a cut $(\tilde{S},\tilde{T})$ with the client nodes $A,B$ satisfying $A\subset \tilde{S}$ and $B\subset \tilde{T}$ respectively. We can regard $A$ and $B$ as source and sink nodes respectively and  we can collect them in a single node as an effective network without affecting the communication capacities between the bipartitions. It is known that the entanglement that can be established  between the source and sink nodes in a tensor network should satisfy \cite{cui2016quantum}
\begin{equation}
    Q(\{A,B\})\leq MC_{\{A,B\}}
\end{equation}
For bi-partite pure entangled states the entanglement that can be extracted $E_D(A,B)$ (entanglement distillation) and the entanglement needed to create the state $E_F$ (the entanglement of formation) are equal and they satisfy 
\begin{align}
     E_D(A,B)&=E_F(A,B)\nonumber\\
     &=S(A)\nonumber\\
     &=rank(\mathrm{Tr}_B(\rho))
\end{align}
Accordingly, for graph states and by using Eq.~\ref{Eq:08}, we have that 
\begin{align}
    Q(\{A,B\})&=\text{rank}_{\mathbb{F}_2}(\Gamma_{AB})\nonumber\\
    &\leq MC_{\{A,B\}}
\end{align}

\end{proof}
\end{theorem}
Although it is obvious, the fact that by distributing EPR pairs between a central node and clients on a star topology, any graph state can be distributed among the clients, can be derived from the previous theorem in  network information theoretical settings.
\begin{corollary}
 The star network topology can distribute any graph state between the leafs of the network from the central node by a single use of the network.
 \begin{proof}
On the one hand, we note that the entanglement rank $r_{AB}$ for any bipartition $\{A,B\}$ of a graph state is bounded from above as:
\begin{equation}
    r_{AB}\leq \min \{|A|,|B|\}
\end{equation}
On the other hand, the min-cut of any bipartition of the clients in the star network topology satisfies:
\begin{equation}
    MC_{\{A,B\}}=\min \{|A|,|B|\}
\end{equation}
Therefore, we have that 
\begin{equation}
    r_{AB}\leq MC_{\{A,B\}}
\end{equation}
and hence, by Theorem.~\ref{Theorem:02} the star topology of the network allows for the distribution of any graph state.
 \end{proof}
 \end{corollary}
Similarly, we can use the results of Theorem.~\ref{Theorem:02} to state the following corollary for GHZ states distribution in arbitrary quantum networks. 

\begin{corollary}
\label{corollary:02}
Any topology of a connected quantum network distributes , by a single shot, a GHZ state between the clients
 \begin{proof}
The proof is straightforward from Theorem.~\ref{Theorem:02} by noting that the entanglement ranks $r_{AB}$ on any bipartition $\{A,B\}$ of the GHZ state satisfy 
\begin{equation}
    r_{AB}= 1
\end{equation}
On the other hand, the min-cut of any bipartition of the clients in any connected network satisfies:
\begin{equation}
    MC_{\{A,B\}}\geq  1
\end{equation}
Therefore, we have that 
\begin{equation}
    r_{AB}\leq MC_{\{A,B\}}
\end{equation}
and hence, by Theorem.~\ref{Theorem:02}, any connected topology of the quantum network can distribute a GHZ state between the clients by single use of the network. 
 \end{proof}
 \end{corollary}

\subsection{Achieveability of the ultimate bounds for entanglement distribution}
\begin{figure}[t]
    \begin{minipage}[t] {.45\textwidth}
        \centering
        \includegraphics[width=\textwidth]{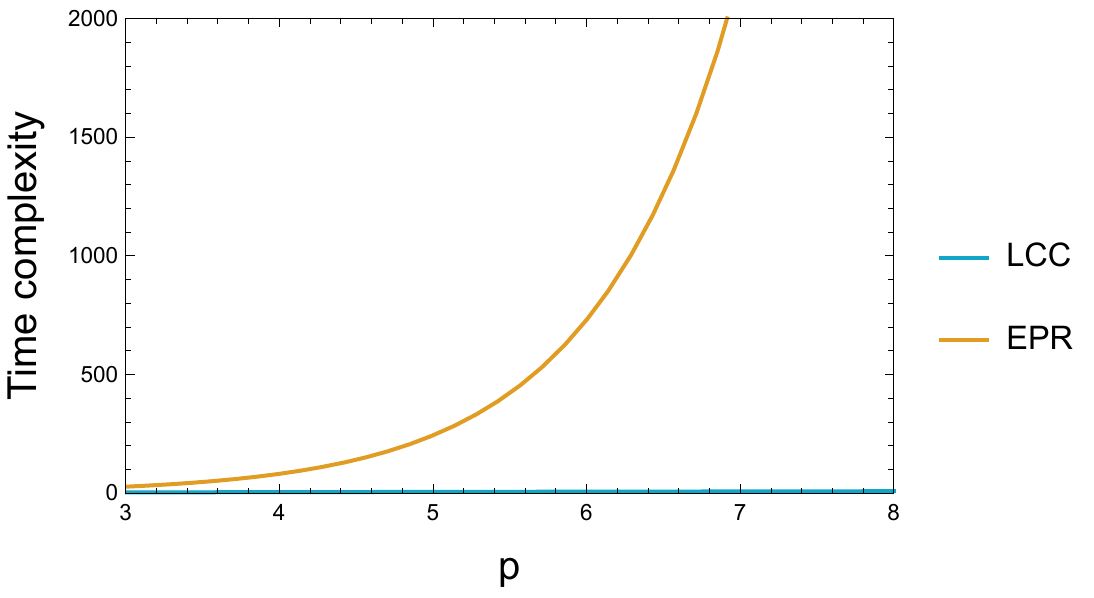}        
        \label{Fig:3-a}
    \end{minipage}
    \hfill
    \begin{minipage}[t] {.45\textwidth}
        \centering
        \includegraphics[width=\textwidth]{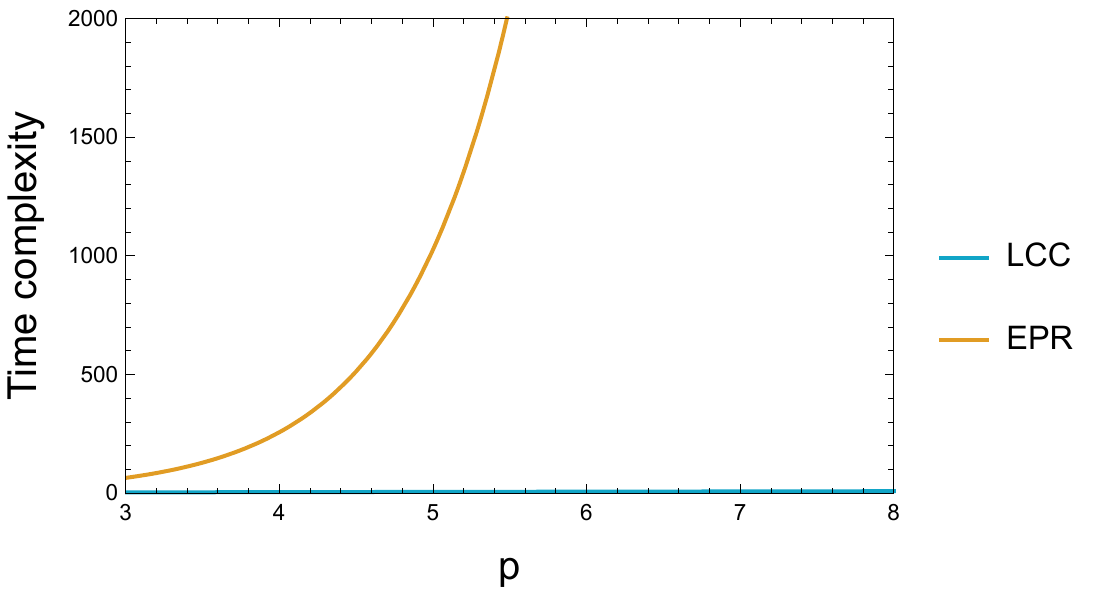}
        \label{Fig:3-b}
        \end{minipage}
        \label{Fig:3}
    \caption{The time complexity of the distribution of a target stabilizer state in different depths of the networks. Left: The time complexity of a network of connectivity $n=3$. Right: The time complexity of a network of connectivity $n=4$.}
\end{figure}

In order to better harness the capabilities of a given quantum network for entanglement distribution, the design of appropriate LQC strategies in relay nodes is needed. This is equivalent to finding the optimal coding strategies in relay nodes in network coding in order to achieve the capacity of a given network. As a matter of fact, our network architecture relying on LQC for entanglement distribution, might be regarded as a quantum network coding  counterpart of the classical network coding paradigm.

As a matter of fact, suppose we endow all the relay nodes in the quantum network with a Linear local quantum coding strategy  (LLQC) given by: 
\begin{align}
    S_x&=U^0_x\otimes \ketbra{0}{0}+U^1_x\otimes \ketbra{1}{1}\nonumber\\
    &=I^{\otimes n}\otimes \ketbra{0}{0} + X^{\otimes n} \otimes \ketbra{1}{1} \label{repetition}
\end{align}
This strategy is equivalent to a quantum repetition code, generating GHZ states in each node. It is easy to notice that this strategy is identical to the isometry in Eq.\ref{eq:18} up to a local unitary given by $I\otimes H^{\otimes n-1}$. Indeed this strategy can only generate entanglement ranks no higher than $r_{AB}=1$, making it not a suitable candidate for achieving the network capacity bound in general.  To illustrate this, we consider different sub-network structures appearing in the network of Fig.~\ref{Fig:09}. For the chain subnetwork (blue nodes), the isometry in Eq.~\ref{repetition} performed in the blue dashed node
\begin{equation}
    S_0=I_0\otimes \ketbra{0}{0}_1 + X_0 \otimes \ketbra{1}{1}_1 \label{epr}
\end{equation}
leads to the same performance of an entanglement swapping. Formally, the effect of the LQC \ref{epr} in the dashed node is equivalent, by Lemma.~\ref{lemma1}, to 
\begin{equation}
    \bra{EPR}_{0}\ket{EPR}_{01}\ket{EPR}_{02}=\ket{EPR}_{12}
\end{equation}
which is exactly entanglement swapping performed at node $0$.
Differently, the performance of the same isometry, in the dashed nodes of sub-networks (orange nodes) and (end nodes), is not equivalent to entanglement swapping, and leads to better entanglement distribution rates as will be discussed in the next section.  This suggests that our multiparty protocol provides an advantage over the Bell pair based approach when the entanglement distribution network has high connectivity.

\section{Advantages}\label{sec:06}
In this section we illustrate the advantages of the use of intermediate quantum coding in quantum networks accordign to different figures of merit.
\subsection{Exponential Latency and Memory Overhead Reduction}
\begin{figure}
    \begin{minipage}[c] {0.49\textwidth}
        \centering
\includegraphics[width=1\columnwidth]{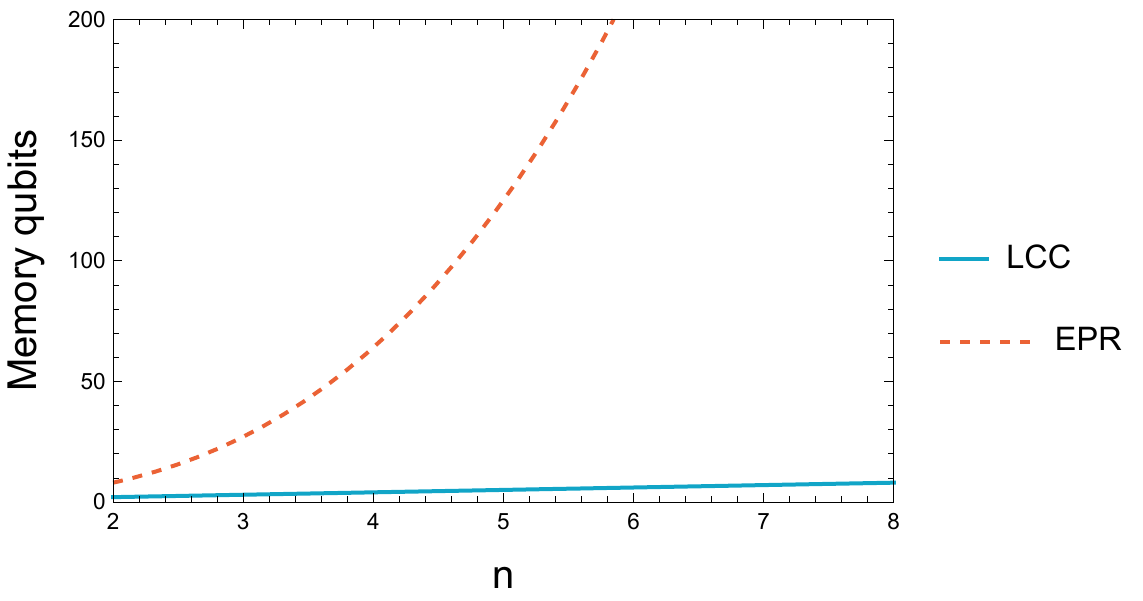}
        \label{Fig:4-a}
    \end{minipage}
    \hspace{0.02\textwidth}
    \begin{minipage}[c] {0.49\textwidth}
        \centering
        \includegraphics[width=1\columnwidth]{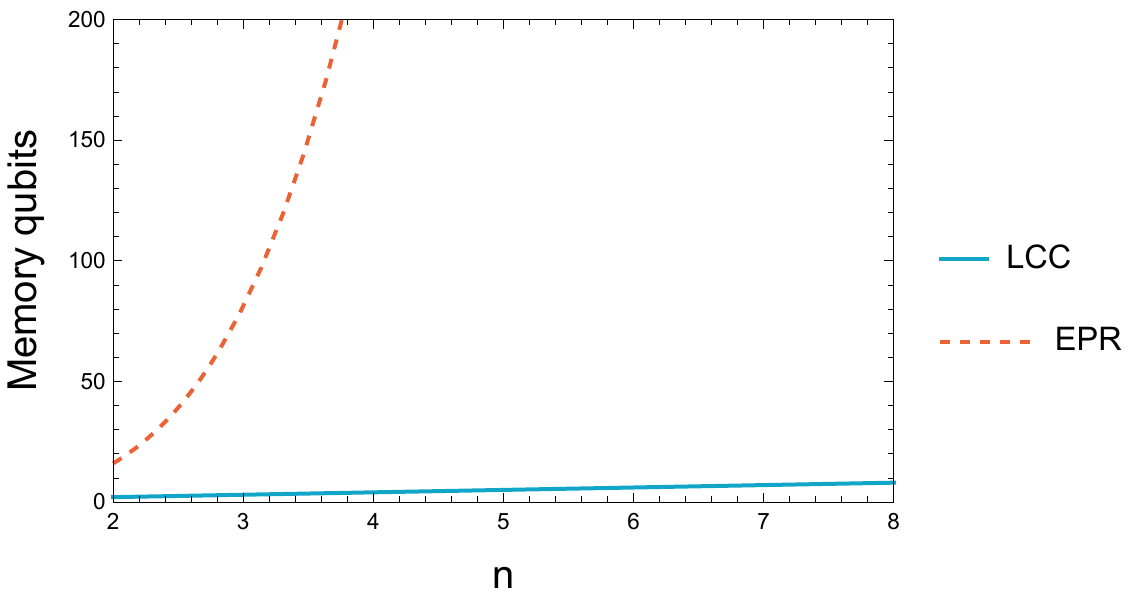}

        \label{Fig:4-b}
        \end{minipage}
    \caption{The number of memory qubits required for the distribution of a target graph states within networks of differents depths. Left: The number of memory qubits required in a network of depth three. Right: The number of memory qubits required in a network of depth four.}
    \label{Fig:4}
\end{figure}
To illustrate these advantages we consider regular networks with different degrees of connectivity. For instance, we will consider a tree network where each node has the same number of children  nodes. We will consider also that the end relay nodes are connected to the same number of clients. 
In such a network the number of relay nodes is given by $N\sim O(n^{p-1}) $, where $p$ is the depth of the network --the smallest distance from the central node and the clients--. Accordingly, the complexities of the single shot protocol with the isometry \ref{repetition} in each relay node,  and the Bell pair based, are given respectively by 
\begin{align}
    &\sim O(p)\nonumber\\
    & \sim O(n^{p})
\end{align}
Clearly, we have an exponential time complexity reduction in the depth of the network for the distribution of GHZ like states.  The results are illustrated in the plots of Fig.~\ref{Fig:3}. We can easily see that the more the deeper the quantum network is the more clear the advantage of using a multipartite distribution strategy is. 

The similar investigation can be conducted for the advantage over the reduction of memory qubits in the network for GHZ. The figure of merit that we use in this case is the maximal number of memory qubits/per node used in the network to achieve the task of distributing a target graph state between the clients. For the LQC approach, the maximum number of memory qubits needed in each relay node is equal to the number of controlled qubits $n^a_c$ required to generate locally the entangled state with an appropriate LQC strategy. For the case of the previous regular network, this is maximized by the number of connectivity $n+1$ of the node according to Lemma.~\ref{lemma1}. Thus, the maximum number of memory qubits $n_m$ satisfies 
\begin{equation}
  \max_{a\in V/C}(n_m)= n^{a^*}_c \leq n+1
\end{equation}
Differently, in the Bell pair-based approach, the highest possible number of memory qubits $n_m$ in each node is of the order of the number of clients in the network 
\begin{equation}
   \max_{a\in V/C}(n_m)\sim n^p 
\end{equation}
The results are illustrated in  Fig.~\ref{Fig:4}. We can notice how the entanglement distribution based on multipartite protocols is beneficial for networks with growing depth and connectivity.

\subsection{Noise resilience}
As a noise model, we assign to every entangled state distributed in the quantum network in Fig.~\ref{Fig:09} a probability of failure  $p$. Namely, the noise model for the entanglement distribution is given by 
\begin{equation}
    \mathcal{N}(\rho)=(1-p)\rho+pe
\end{equation}
where $e$ is the error flag which will abort the distribution protocol. 
Accordingly, the error propagation in the network for the distribution of $n$-EPR pairs would be given by 
\begin{equation}
    \mathcal{N}^{\otimes n} (\rho^{\otimes n})=\sum_{k=0}^n {n \choose k} (1-p)^{n-k}p^k\rho^{\otimes n-k}\otimes e^{\otimes k} 
\end{equation}

Clearely, the probability that the protocol does not abort is when $k=0$ which is given by 
\begin{equation}
    p_{sucess}=(1-p)^n
\end{equation}
As we can easily notice, the less qubit channels are used to establish the target state between the clients the less noisy is the protocol. 

For the case of the tree network of depth $p=2$, our protocol outperforms the protocol relying on central node and EPR pairs distribution. The performance of the two protocols is illustrated in Fig.~\ref{Fig:10-b}. It is clear that for different target graph states we have
\begin{equation}
    p_{sucess}^{LQC}>p_{sucess}^{EPR}
\end{equation}
showing that our protocol is more robust to noise than other protocols. 

\begin{figure}
\begin{center}
 \includegraphics[width=0.8\columnwidth]{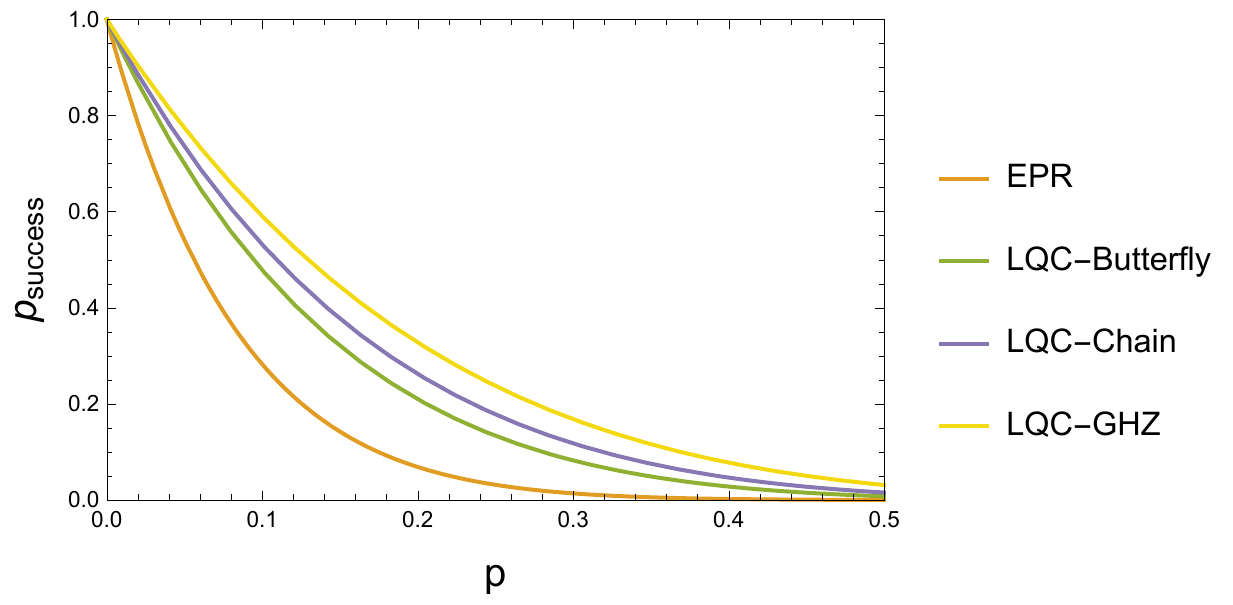}
\end{center}
        
        \caption{A plot comparing the probability of success of the distribution of different graph states in the framework of the paper with the EPR based approach.}
        \label{Fig:10-b}  
\end{figure}

\subsection{Distributed Quantum Storage}

A distributed storage system is an infrastructure that can split data across multiple physical servers, and often across more than one data center. It typically takes the form of a cluster of storage units. The storage system includes features that serves as a mechanism for data synchronization and coordination between cluster storage nodes that are geographically dispersed. Moreover, it has the intelligence to detect and respond to failures and cyber attacks. The objective is to achieve very low latency by storing data physically near the location it will be used. 

Fundamentally, the transition from classical distributed storage to quantum distributed storage of quantum data, cannot be smooth due to the no-cloning theorem. Although, our distributed stabilizer states generation discussed in Sec.~\ref{sec:04}, using intermediate quantum encoding, can be harnessed to enhance quantum storage of quantum data significantly. 

Consider the network in Fig.~\ref{fig:05}, where the bulk nodes (orange) denoted by $\{b_i\}$ want to store their single qubit data $\{q_{b_i}\}$ into the neighbooring cluster nodes (red) denoted as $\{c_ij\}$. The best that each node can do to store its qubit, independentely of the other bulk nodes, is to encode it in a three qubit code given by a GHZ-like state 
\begin{equation}
    \ket{GHZ}^{b_i}=\frac{1}{\sqrt{2}}\big(\ket{000}^{c_{i1}c_{i2}c_{i3}}+\ket{111}^{c_{i1}c_{i2}c_{i3}}\big)
\end{equation}
It can easily be noticed that by encoding in such codes, single qubit errors that might occur during the storage cannot be corrected by no means, making the encoding vulnerable to noise and adversarial attacks in the transmission. 
Differently, if the bulk nodes (orange) collaborates and look for local codes that harness the full capacity of the network as discussed in Sec.~\ref{sec:05}, they might achieve better performance. to do so, let the bulk nodes choose to encode their single qubits in a five qubit code each. The five qubit code, denoted as $[[5,1,3]]$, can be described by a cyclic five qubits graph state, which have maximal entanglement ranks among five qubit graph states. Equivalently, it might be described by the stabilizer generators given by\cite{lidar2013quantum,gottesman1997stabilizer}: 
\begin{align}
    XZZXI, IXZZX, XIXZZ, ZXIXZ
\end{align}
This  can encode at most a single logical qubit, and can correct single qubit errors and up to two erasure errors.
\begin{corollary}
\label{corollary:03}
 By contracting the three $[[5,1,3]]$ codes with Bell pairs corresponding to edges according to Eq.~\ref{Eq:theorem 4}, a nine qubit code $[[9,3,3]]$ is established between remote cluster nodes. This nine qubit code has as generators
\begin{align}
    &XZZXIXZZX, XIXXZZXZZ, IXZZXIIXZ\nonumber\\
    &ZXIIXZZXI, ZYYXZZIII, YXXYXXIII
\end{align}  
\end{corollary}
It can be simply checked that these nine qubit Pauli operators are commuting. Moreover, it is shown in Appendix.~\ref{app.1} that they are independent of each other. Therefore, they provide a stabilizer code storing three logical qubits -encoded locally in the three bulk nodes-, in nine physical qubits-shared between the remote cluster nodes-.  Furthermore, it is shown in details in the  Appendix that the established nine qubit code can correct for single qubit errors  and up to two erasure errors. As a matter of fact, the remotely established code between remote cluster nodes, and storing the three qubits from the bulk nodes is denoted by $[[9,3,3]]$. Not surprisingly, the obtained code satisfies, indeed, the quantum singleton bound given by\cite{lidar2013quantum,gottesman1997stabilizer}:

\begin{equation}
    k\leq n-2(d-1)
\end{equation}

We can generalize the previous example of quantum storage to regular networks with a fixed number of neighbors. As a matter of fact we have the following proposition:
\begin{figure}
\begin{center}
\includegraphics[width=0.9\columnwidth]{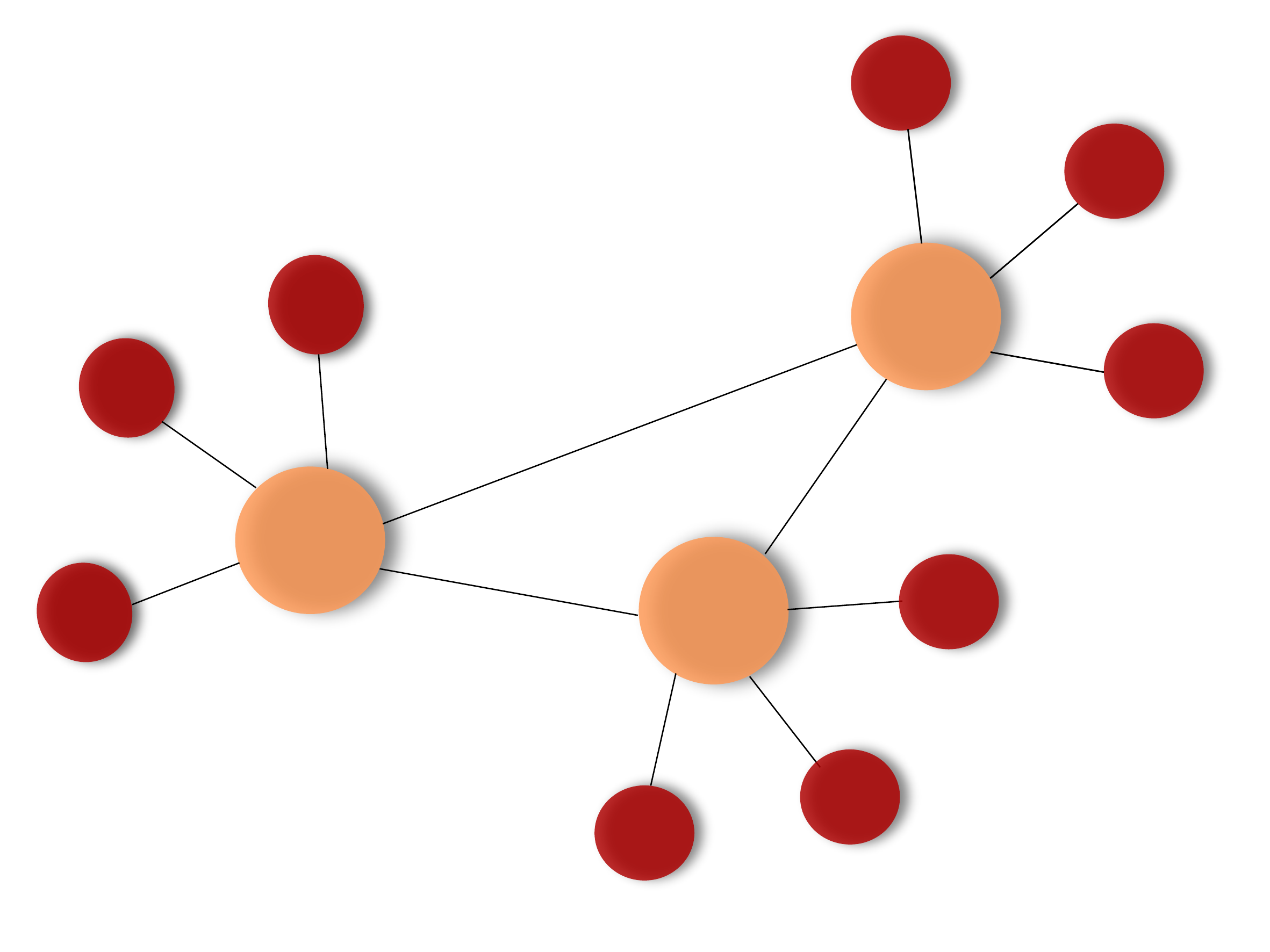}
\end{center}
    \caption{A network where the bulk nodes (orange) store a single qubit quantum data and wants to store it in its neighbooring cluster  nodes (red)}
        \label{fig:05}
\end{figure}
\begin{proposition}
\label{proposition:01}
Let a quantum storage network where the boundary $B$, containing the quantum registers, and a bulk containing $m$ nodes each of which encodes $k$ logical qubits in a stabilizer code $[[l,k,d]]$ of $l$ physical qubits and distance $d$. Then, by contracting the different codes, a stabilizer code $[[|B|,mk,D]]$, encoding the $mk$ qubits, can be established between the quantum registers $|B|$, where its distance $D$ satisfies:
\begin{equation}
    D\leq \frac{|B|+ml-2m(d-1)-2km}{2}+1
\end{equation}
 \end{proposition}
 \begin{proof}
 Suppose we have a network with bulk nodes having $k$ qubits encoded in a stabilizer quantum error correcting code $[[l,k,d]]$, where $l$ is the number of neighboring nodes. The code-space of such code is given by 
\begin{equation}
    \ketbra{\psi}{\psi}_i=\frac{1}{2^{l-k}}\sum_{b_i} u_{b_i} 
\end{equation}
where $\{u_{b_i}\}$ are the elements of the corresponding stabilizer group of node $i$. Suppose also that we have $m$ nodes in the bulk of the network and we have the storing registers sitting at the boundary $B$ of the network. According to Theorem.~\ref{Theorem:04} the state on the boundary would be given by the projector
\begin{align}
    \ketbra{\Psi}{\Psi}_B&=\mathrm{Tr}(\otimes_{i=1}^m\ketbra{\psi}{\psi}_i\ketbra{e}{e}^{\otimes (lm+|B|)})\nonumber\\
    &=\frac{1 }{2^{(n-k)m+lm+|B|}}\sum_{s_{AB}}\sum_{t_{A}} \mathrm{Tr}(s_{AB}t_A)\nonumber\\
    &=\frac{1}{2^{(|B|-km)}}\sum_{u_{B}}u_{B}
\end{align}
The last two equalities follow from the derivation of Theorem.~\ref{Theorem:04}. with $km$ is the number of logical qubits stored into the $|B|$ quantum registers. Additionally, according to the quantum singleton bound, the best distance $D$, i.e., minimum Pauli error weight that cannot be corrected, that this stabilizer code can have should satisfy:
\begin{equation}
\label{ineq}
    km\leq |B|-2(D-1)
\end{equation}
therefore:
\begin{equation}
    D\leq \frac{|B|-km}{2}+1
\end{equation}
Differently, the single codes sitting in the bulk satisfy 
\begin{equation}
\label{ineq1}
    km\leq ml-2m(d-1)
\end{equation}
by adding the two inequalities \ref{ineq} and \ref{ineq1} we get:
\begin{equation}
    2km\leq |B|-2(D-1)+ml-2m(d-1)
\end{equation}
and hence, $D$ should satisfy:
\begin{equation}
    D\leq \frac{|B|+ml-2m(d-1)-2km}{2}+1
\end{equation}
\end{proof}

Proposition.~\ref{proposition:01} can be regarded as a straightforward result of known entropic inequalities. To better understand this, let each of the $m$ nodes containing $K$ EPR pairs where half of these pairs are encoded into a $l$-qubits quantum code with distance $d$ and the other half is kept as a purifying system $R_i$. We can split the encoded $l$ qubits part into three partitions $A_iF_iC_i$ where $\dim(A_i)=\dim(F_i)=2^{d-1}$ and $\dim(C_i)=2^{l-2(d-1)}$. The global system is thus given by $RAFC=R_i^{\otimes_{i=1}^m}A_i^{\otimes_{i=1}^m}F_i^{\otimes_{i=1}^m}C_i^{\otimes_{i=1}^m}$. In the one hand, it is known that as long as the quantum code has distance $d$, any partition of $d-1$ qubits of the codespace is independent of the encoded state \cite{lidar2013quantum,gottesman1997stabilizer}. Formally, we should have
\begin{align}
    &S(R_iA_i)=S(R_i)+S(A_i)\nonumber\\
    &S(R_iF_i)=S(R_i)+S(F_i)\nonumber\\
    &S(R)=S(AFC)=\sum_{i=1}^mS(R_i)=mk\nonumber\\
\end{align}
Accordignly, we can have
\begin{align}
        &S(R)=S(FC)-S(A)=\sum_{i=1}^mS(R_i)=mk\nonumber\\
        &S(R)=S(AC)-S(F)=\sum_{i=1}^mS(R_i)=mk
\end{align}
Applying the subbaditivity property of the entropy to the last two inequalities we get
\begin{align}
    &mk=S(R)\leq S(C)+S(F)-S(A)\nonumber\\
    &mk=S(R)\leq S(C)+S(A)-S(F)
\end{align}
Therefore
\begin{align}
    mk&\leq S(C)\nonumber\\
    & = \sum_{i=1}^m S(C_i)\nonumber\\
    &\leq m(l-2(d-1))
\end{align}

In the pther hand, we note that the contraction of the codes maps the entanglement in the bulk into the boundary $B$ as $RAFC\rightarrow RB$ with 
\begin{equation}
    S(R)=\min\{mk,|B|\}
\end{equation}
Provided that $|B|>km$ and that the state on the boundary is indeed a quantum code according to Theorem.~\ref{Theorem:04}, and in the same spirit as before we have:
\begin{align}
    S(R)=mk&\leq S(C_B)\nonumber\\
    &\leq \dim(C_B)
\end{align}
where $\dim(C_B)=|B|-2(D-1)$. 

We can easily verify that the example of the Network in Figure.~\ref{fig:05}, with $|B|=9$, $l=5$, $m=3$, $d=3$ and $k=1$ satisfy Proposition~\ref{proposition:01} in the sense that we have the singleton bound $D=3\leq 4$, which is the smallest distance code for $|B|=9$ physical quantum registers and encoding $k=3$ logical qubits.

Clearly, the collective efforts of the bulk nodes -by harnessing the full capacity of the network- to establish a common code for distributed storage is advantageous than encoding separately in a centralized way.  This gives more resilience to errors that might occur during the storage time or during the transmission process, in the meanwhile, it provides more resilience to losses, or to dishonest participation, when the stored information is to be retrieved. Therefore, this shows the importance of choosing appropriate quantum network encoding for providing secure quantum cryptographical schemes for distributively storing quantum data.
Indeed, this would not be possible if the network relies solely on bipartite entanglement distribution and bipartite entanglement swapping. 


\section{Conclusions and Future Work}\label{sec:07}
A new approach for entanglement distribution in quantum networks has been presented. The approach relies on linear quantum coding strategies in relay nodes, which is an inherently multiparty strategy. We have provided the usefulness of this approach in terms of single shot capacities of the quantum network, making it analogous to the classical network coding paradigm. Moreover, a performance analysis has been elaborated for particular ordered network topologies. We have shown that for some networks, the time complexity that the LQC approach achieves, is exponentially reduced --in the depth parameter of the network-- with respect to the time complexity that the usual approach relying on EPR pairs distribution with a central node gives. Similarly, a polynomial advantage --in the connectivity parameter- on the maximal number of memory qubits needed throughout the network is observed. These figure of merits, illustrates the advantage of the LQC paradigm in practical communication networks with high depth and high connectivity. Additionally, a resilience to noise has been illustrated for some specific model, where the probability of success of distributing some target graph states, using LQC, exceeds with significant amount its counterpart in the EPR based approach. As an application we have shown how the proposed protocol when uses quantum error correcting codes as LQC can be useful for distributed quantum storage.  Indeed, all the advantages of LQC were assuming that the distribution network is given. Although, in real case scenarios, the distribution network should, of course, be determined. 

\bibliographystyle{IEEEtran}
\bibliography{sn-bibliography}

\appendices

\section{Proof of Corollary.~3}
\label{app.1}
\begin{figure*}[t]
\normalsize
\begin{equation}
\label{matrix}
H=\begin{bmatrix}[ccccccccc|ccccccccc]
   1&0&0&1&0&1&0&0&1&0&1&1&0&0&0&1&1&0 \\
   1&0&1&1&0&0&1&0&0&0&0&0&0&1&1&0&1&1 \\
   0&1&0&0&1&0&0&1&0&0&0&1&1&0&0&0&0&1 \\
   0&1&0&0&1&0&0&1&0&1&0&0&0&0&1&1&0&0\\
   0&1&1&1&0&0&0&0&0&1&1&1&0&1&1&0&0&0\\
   1&1&1&1&1&1&0&0&0&1&0&0&1&0&0&0&0&0\\
\end{bmatrix}
\end{equation}
\hrulefill
\end{figure*}
The proof of Corollary.~\ref{corollary:03} requires explicit computation of the contractions of the three five qubit codes $[[5,1,3]]$ with EPR pairs. Without any loss of generality, we let the EPR pairs used for 
contraction to be given by 
\begin{equation}
    \ket{e}=\frac{1}{\sqrt{2}}(\ket{00}+\ket{11})
\end{equation}
This is indeed a graph state with stabilizer group given explicitly by the elements:
\begin{equation}
\label{stab-group}
   II,\quad XZ,\quad ZX, \quad -YY
\end{equation}
Since $3$ EPR pairs are needed for contraction withing the bulk nodes, the overall state for contraction is $\ket{e}^{\otimes 3}$ with stabilizer group with elements given by the tensor product of the individual elements of the stabilizer group in \ref{stab-group}. The cardinality of this group is indeed given by $4^3=64$ elements. We don't count all the possibilities of combinations of the elements, but we give an example that illustrates the procedure. If we allow the nodes to be contracted by the stabilizer elements: 
\begin{equation}
    ZX,\quad XZ,\quad XZ 
\end{equation}
on each of the edges of the bulk, then the global stabilizer element of the contraction is given by
\begin{equation}
\label{exemp}
   \underbrace{XZ}\underbrace{ZX}\underbrace{XZ}
\end{equation}
where the under-braces refer to operators that should act on different bulk nodes.
Moreover, the elements of stabilizer group of the code $[[5,1,3]]$ are given by:
\begin{align}
\label{stab-group-5}
    & ZXI\underbrace{XZ},\quad ZZX\underbrace{IX},\quad XZZ\underbrace{XI},\quad IXZ\underbrace{ZX}\nonumber\\
    &-IYX\underbrace{XY},\quad -YYZ\underbrace{IZ},\quad -ZIZ\underbrace{YY},\quad -YIY\underbrace{XX}\nonumber\\
    &-ZYY\underbrace{ZI},\quad -XYI\underbrace{YX},\quad -XXY\underbrace{IY},\quad -IZY\underbrace{YZ}\nonumber\\
    &-YZI\underbrace{ZY},\quad -YXX\underbrace{YI},\quad XIX\underbrace{ZZ},\quad III\underbrace{II}
\end{align}
The underbraces refer to the parts of the codes to be contracted with the neighboring bulk nodes. There are two neighboring bulk nodes to each bulk node, therefore, two physical qubits have to be contracted. 
If we contract \ref{exemp} with the stabilizer group of the three individual codes \ref{stab-group-5} the following nine qubits residual element remains:
\begin{equation}
    ZXIIXZZXI
\end{equation}
In fact, we have $64$ nine qubit residual elements after full contraction. They indeed form a group as we have shown in the proof of Theorem.~\ref{Theorem:04}. In order to prove that this residual group  encodes three qubits, we have to show that it contains six independent generators. Clearely, proving that the six elements 
\begin{align}
\label{genera}
    &XZZXIXZZX, XIXXZZXZZ, IXZZXIIXZ\nonumber\\
    &ZXIIXZZXI, ZYYXZZIII, YXXYXXIII
\end{align}  
belonging to the residual group, are independent, would be sufficient. To do so, we switch to the symplectic representation  of the stabilizer elements \cite{gottesman1997stabilizer,lidar2013quantum}. Accordignly, the above elements might be represented by the matrix $H$ given by \ref{matrix}.
It is clear that the last two rows of the matrix $H$ cannot be written as any possible linear combination over $GF(2)$ of the other four rows. In the meanwhile the last two rows are independent from each other. From the last three entries of the $Z$ part of the matrix $H$ we can notice that the first four rows are linearly independent. Therefore, the rows of the matrix $H$ are linearly independent of each other. Hence, the corresponding stabilizer elements in \ref{genera} are indeed generators of the residual stabilizer group.

Now we should prove that the nine qubit code generated by the stabilizer elements in \ref{genera} has distance $3$. To do this, we should show that any $9$-qubit Pauli operator of weight two should anticommute with at least one of the generators \ref{genera}. Let $g_l$ be a generator of the residual nine qubit stabilizer group, and let $O_B$ be a two qubit Pauli acting on the nine qubits. Form the proof of Theorem.~\ref{Theorem:04} the following is true
\begin{equation}
    g_l=\mathrm{Tr}_A(s_{BA}t_A)
\end{equation}
for some $s_{BA}$, an element of the stabilizer group of the three five qubit codes and $t_A$ an element of the stabilizer group corresponding to the contracting edges. Accordingly, we have

\begin{align}
O_Bg_l&=\mathrm{Tr}_A(O_B\otimes I_A)\mathrm{Tr}_A(s_{BA}t_A)=\mathrm{Tr}_A(O_B\otimes t_As_{BA})         
\end{align}

Since $\{s_{BA}\}$ is a stabilizer group correcting for at most two Pauli errors, then there is at least one element $s_{BA}$ that satisfies

\begin{equation}
    \{O_B\otimes I_A, s_{BA}\}=0
\end{equation}
Hence, there is at least one generator $g_l$ of the residual stabilizer group satisfying 
\begin{equation}
    \{O_B, g_l\}=0
\end{equation}
Therefore, the nine qubit code can correct for any single qubit error. The same argument applies to show that the distance of the code is indeed $D=3$.

\end{document}